\documentclass[12pt]{article}

\usepackage{amsthm}
\usepackage{amsmath,amssymb}
\usepackage{lmodern}
\usepackage[T1]{fontenc}
\usepackage{microtype}
\usepackage[letterpaper, margin=1in]{geometry}
\usepackage{hyperref}
\usepackage{multirow}
\usepackage{tabularx}
\usepackage{colortbl}
\usepackage{color}
\usepackage{graphicx}
\usepackage[small]{caption}
\usepackage{setspace}

\newtheorem{theorem}{Theorem}[section]
\newtheorem{lemma}[theorem]{Lemma}

\newtheorem{corollary}[theorem]{Corollary}

\newtheorem{definition}[theorem]{Definition}

\newcommand{\zo}{\{0,1\}}

\newcommand{\eps}{\varepsilon}

\newcommand{\E}[2]{{\mathbb{E}_{#1}\left[#2\right]}}

\newcommand{\Zbb}{\mathbb{Z}}
\newcommand{\aanote}[1]{}
\newcommand{\irnote}[1]{}

\title{Tight Lower Bounds for Data-Dependent Locality-Sensitive Hashing}
\author{
Alexandr Andoni\footnote{Work done in part while the author was at
  the Simons Institute for the Theory of Computing, Berkeley University.}\\Columbia University
\and
Ilya Razenshteyn\\CSAIL MIT}

\begin{document}
\maketitle

\begin{abstract}
We prove a tight lower bound for the exponent $\rho$ for
data-dependent Locality-Sensitive Hashing schemes, recently used to
design efficient solutions for the $c$-approximate nearest neighbor
search. In particular, our lower bound matches the bound of $\rho\le
\frac{1}{2c-1}+o(1)$ for the $\ell_1$ space, obtained via the recent
algorithm from [Andoni-Razenshteyn, STOC'15].

In recent years it emerged that data-dependent hashing is strictly
superior to the classical Locality-Sensitive Hashing, when
the hash function is data-{\em independent}. In the latter setting, the best
exponent has been already known: for the $\ell_1$ space, the tight
bound is $\rho=1/c$, with the upper bound from [Indyk-Motwani,
  STOC'98] and the matching lower bound from [O'Donnell-Wu-Zhou,
  ITCS'11].

We prove that, even if the hashing is data-dependent, it must hold
that $\rho\ge \frac{1}{2c-1}-o(1)$. To prove the result, we need to
formalize the exact notion of data-dependent hashing that also
captures the complexity of the hash functions (in addition to their
collision properties). Without restricting such complexity, we would
allow for obviously infeasible solutions such as the Voronoi diagram
of a dataset.  To preclude such solutions, we require our hash functions
to be succinct.
This condition is
satisfied by all the known algorithmic results.

\end{abstract}
\onehalfspacing
\thispagestyle{empty}
\newpage
\setcounter{page}{1}

\section{Introduction}

We study lower bounds for the high-dimensional nearest neighbor search
problem, which is a problem of major importance in several areas, such
as databases, data mining, information retrieval, computer vision,
computational geometry, signal processing, etc. This problem suffers
from the ``curse of dimensionality'' phenomenon: either space or
query time are exponential in the dimension $d$. To escape this curse,
researchers proposed {\em approximation} algorithms for the problem.
In the $(c, r)$-approximate near neighbor problem, the data structure
may return any data point whose distance from the query is at most $c
r$, for an approximation factor $c > 1$ (provided that there exists a
data point within distance $r$ from the query).  Many approximation
algorithms are known for this problem: e.g., see
surveys~\cite{samet-new, AI-CACM, a-nnson-09, wssj-hsss-14}.

An influential algorithmic technique for the approximate near
neighbor search (ANN) is the {\em Locality Sensitive Hashing} (LSH)
\cite{im-anntr-98,him-anntr-12}. The main idea is to hash the points
so that the probability of collision is much higher for points that
are close to each other (at distance $\le r$) than for those which are
far apart (at distance $> cr$). Given such hash functions, one can
retrieve near neighbors by hashing the query point and retrieving
elements stored in buckets containing that point.  If the probability
of collision is at least $p_1$ for the close points and at most $p_2$
for the far points, the algorithm solves the $(c, r)$-ANN using
essentially $O(n^{1+\rho}/p_1)$ extra space and $O(dn^{\rho}/p_1)$
query time, where
$\rho=\tfrac{\log(1/p_1)}{\log(1/p_2)}$~\cite{him-anntr-12}. The value of the
exponent $\rho$ thus determines the ``quality'' of the LSH families
used.

Consequently, a lot of work focused on understanding the best possible
value $\rho$ for LSH, including the sequence of upper bounds
\cite{im-anntr-98, diim-lshsp-04, ai-nohaa-06} and lower bounds
\cite{mnp-lblsh-07, owz-olbls-11}. Overall, they established the
precise bounds for the best value of $\rho$: for~$\ell_1$ the tight
bound is $\rho=\tfrac{1}{c}\pm o(1)$. In general, for $\ell_p$, where $1\le
p\le2$, the tight bound is $\rho=\tfrac{1}{c^p}\pm o(1)$.

Surprisingly, it turns out there exist {\em more efficient} ANN data
structures, which step outside the LSH framework. Specifically,
\cite{ainr-blsh-14, ar-optimal-15} design algorithms using the concept
of {\em data-dependent hashing}, which is a randomized hash family
that itself adapts to the actual given dataset. In particular, the
result of \cite{ar-optimal-15} obtains an exponent
$\rho=\tfrac{1}{2c^p-1}+o(1)$ for the $\ell_p$ space, thus improving
upon the best possible LSH exponent essentially by a factor of 2 for
both $\ell_1$ and $\ell_2$ spaces.

\paragraph{Our result.}
Here we prove that the exponent $\rho=\tfrac{1}{2c^p-1}$ from
\cite{ar-optimal-15} is essentially optimal even for data-dependent
hashing, and cannot be improved upon. Stating the precise theorem
requires introducing the precise model for the lower bound, which we
accomplish below. For now, we state our main theorem informally:

\begin{theorem}[Main, informal]
\label{thm:mainInformal}
Any data-dependent hashing scheme for $\ell_1$ that achieves
probabilities $p_1$ and $p_2$ must satisfy
$$
\rho = \frac{\log 1 / p_1}{\log 1 / p_2} \ge \frac{1}{2c-1}-o(1),
$$
as long as the description
complexity of the hash functions is sufficiently small.

An immediate consequence is that $\rho\ge \frac{1}{2c^p-1}-o(1)$ for
all $\ell_p$ with $1\le p\le 2$, using the embedding
from \cite{llr-ggsaa-95}.
\end{theorem}

\subsection{Lower Bound Model}
\label{sec:model}

To state the precise theorem, we need to formally describe what is
data-dependent hashing. First, we state the definition of
(data-independent) LSH, as well as define LSH for a fixed
dataset~$P$.

\begin{definition}[Locality-Sensitive Hashing]
  We say that a hash family $\mathcal{H}$ over $\{0, 1\}^d$ is {\em $(r_1, r_2, p_1, p_2)$-sensitive}, if for every
  $u, v \in \{0, 1\}^d$ one has:
  \begin{itemize}
  \item if $\|u - v\|_1 \leq r_1$, then $\underset{h \sim \mathcal{H}}{\mathrm{Pr}}[h(u) = h(v)] \geq p_1$;
  \item if $\|u - v\|_1 > r_2$, then $\underset{h \sim \mathcal{H}}{\mathrm{Pr}}[h(u) = h(v)] \leq p_2$.
  \end{itemize}
\end{definition}

We now refine the notion of data-independent LSH, where we require the
distribution to work only for a particular dataset $P$.

\begin{definition}[LSH for a dataset $P$]
  \label{data_dependent_def}
  A hash family $\mathcal{H}$ over $\{0, 1\}^d$ is said to be {\em $(r_1, r_2, p_1, p_2)$-sensitive for a dataset} $P \subseteq \{0, 1\}^d$, if:
  \begin{itemize}
  \item for every $v \in \{0, 1\}^d$ and every $u \in P$ with $\|u - v\|_1 \leq r_1$ one has
    $$
    \underset{h \sim \mathcal{H}}{\mathrm{Pr}}[h(u) = h(v)] \geq p_1;
    $$
  \item
    $
    \underset{\substack{h \sim \mathcal{H}\\u, v \sim P}}{\mathrm{Pr}}[h(u) = h(v) \mbox{ and }\|u - v\|_1 > r_2] \leq p_2.
    $
  \end{itemize}
\end{definition}

Note that the second definition is less stringent than the first one:
in fact, if all the points in a dataset are at distance more than
$r_2$ from each other, then an LSH family $\cal H$ is also LSH for
$P$, but not necessarily vice versa! Furthermore, in the second
definition, we require the second property to hold only
{\em on average} (in contrast to {\em every} point as in the first
definition). This aspect means that,
while Definition~\ref{data_dependent_def} is certainly necessary for
an ANN data structure, it is not obviously \emph{sufficient}. Indeed,
the algorithm from \cite{ar-optimal-15} requires proving additional
properties of their partitioning scheme, and in particular analyzes
\emph{triples} of points. Since here we focus on lower bounds, this
aspect will not be important.

We are now ready to introduce data-dependent hashing.

\paragraph{Data-dependent hashing.}
Intuitively, a data-dependent hashing scheme is one where we can pick
the family $\cal H$ {\em as a function of $P$}, and thus be sensitive
for $P$. An obvious solution
would hence be to choose $\cal H$ to consist of a single hash function
$h$ which is just the Voronoi diagram of the dataset $P$: it will be
$(r, cr, 1, 0)$-sensitive for $P$ and hence $\rho=0$. However this
does not yield a good data strcture for ANN since {\em evaluating} such a hash
function on a query point $q$ is as hard as the original problem!

Hence, ideally, our lower bound model would require that the hash function
is {\em computationally efficient} to evaluate. We do not know how to
formulate such a condition which would not make the question as hard
as circuit lower bounds or high cell-probe lower bounds, which would
be well beyond the scope of this paper.

Instead, we introduce a condition on the hash family
that can be roughly summarized as ``hash functions from the family are succinct''.
For a precise definition and discussion see below. For now, let us point out that
all the known algorithmic results satisfy this condition.

Finally, we are ready to state the main result formally.

\begin{theorem}[Main theorem, full]
  \label{main_thm}
Fix the approximation $c > 1$ to be a constant. Suppose the dataset
$P$ lives in the Hamming space $\zo^d$, where the dataset size $n=|P|$
is such that $d = \omega(\log n)$ as $n$ tends to infinity.  There exist
distance thresholds $r$ and $(c-o(1))r$ with the following property.

Suppose there exist $T$ hash functions $\{h_i\}_{1 \leq i \leq T}$
over $\{0, 1\}^d$ such that for every $n$-point dataset $P$ there
exists a distribution ${\cal H}_P$ over $h_i$'s such that the
corresponding hash family is $(r, (c - o(1))r, p_1, p_2)$-sensitive
for $P$, where $0 < p_1, p_2 < 0.99$.
For any such data-dependent scheme it must either hold that
$
\rho=\frac{\log 1/p_1}{\log 1/p_2}\ge \frac{1}{2c - 1} - o(1)
$
or that
$
        \frac{\log T}{p_1} \geq n^{1 - o(1)}
$.
\end{theorem}

\paragraph{Interpreting Theorem~\ref{main_thm}}
Let us explain the conditions and the conclusions of Theorem~\ref{main_thm} in more detail.

We start by interpreting the conclusions. As explained above, the bound
$\rho = \frac{\log 1 / p_1}{\log 1 / p_2} \geq \frac{1}{2c - 1} - o(1)$ directly implies the lower bound
on the query time $n^{\frac{1}{2c - 1} - o(1)}$ for any scheme that is based on data-dependent hashing.

The second bound $\frac{\log T}{p_1} \geq n^{1 - o(1)}$ is a little bit more mysterious. Let us now explain
what it means precisely. The quantity $\log T$ can be interpreted as the \emph{description complexity} of a hash function
sampled from the family. At the same time, if we use a family with collision probability $p_1$ for close points, we need
at least $1 / p_1$ hash tables to achieve constant probability of success. Since in each hash table we evaluate at least one hash function, the quantity $\frac{\log T}{p_1}$ can be interpreted as the lower bound for the
\emph{total space occupied by hash functions we evaluate during each query}.
In all known constructions of (data-independent or data-dependent) LSH families~\cite{im-anntr-98, diim-lshsp-04,
ai-nohaa-06, TT07, ainr-blsh-14, ar-optimal-15} the \emph{evaluation time} of a single hash function is comparable to the space it occupies
(for discussion regarding why is it true for~\cite{ar-optimal-15}, see Appendix~\ref{apx:upper}), thus, under this assumption, we can not achieve query time $n^{1 - \Omega(1)}$, unless $\rho \geq \frac{1}{2c - 1} - o(1)$. On the other hand, we can achieve $\rho = 0$ by considering a data-dependent hash family that consists only of the Voronoi diagram of a dataset (trivially, $p_1 = 1$ and $p_2 = 0$ for this case), thus the conclusion
$\frac{\log T}{p_1} \geq n^{1 - o(1)}$ can not be omitted in general\footnote{For the Voronoi diagram, $\log T \geq n$, since
to specify it, one needs at least $n$ bits.}. Note that the ``Voronoi diagram family'' is very slow to evaluate: to locate a point
we need to solve an instance of exact Nearest Neighbor Search, that is unlikely to be possible to do in strongly sublinear time. Thus,
this family satisfies the above assumption ``evaluation time is comparable to the space''.

We now turn to interpreting the conditions.
We require that $d = \omega(\log n)$\footnote{When ANN for the general dimension $d$ is being solved, one usually first performs some form of the
dimensionality reduction~\cite{jl-elmhs-84,kor-esann-00,dg-eptjl-03}. Since at this stage we do not want distances to be distorted by a factor more than $1 + o(1)$, the target
dimension is precisely $\omega(\log n)$. So, the assumption $d = \omega(\log n)$ in Theorem~\ref{main_thm} in some sense captures
a \emph{truly high-dimensional case}.}. We conjecture that this requirement is necessary and for $d = O(\log n)$
there is an LSH family that gives a better value of $\rho$ (the improvement, of course, would depend on the hidden constant in
the expression $d = O(\log n)$). Moreover, if one steps outside the pure data-dependent LSH framework, in a recent paper~\cite{bdgl15} an improved data structure for ANN for the case
$d = O(\log n)$ is presented, which achieves an improvement similar to what we conjectured above.

\subsection{Techniques and Related Work}

There are two components to our lower bound, i.e., Theorem \ref{main_thm}.

The first component is a lower bound for {\em data-independent} LSH
for a {\em random dataset}. We show that in this case, we must have
$\rho\ge \tfrac{1}{2c-1} - o(1)$. This is in contrast to the lower bound of
\cite{owz-olbls-11}, who achieve a higher lower bound but for the case
when the (far) points are correlated. Our lower bound is closer in
spirit to \cite{mnp-lblsh-07}, who also consider the case when the far
points are random uncorrelated. In fact, this component is a
strengthening of the lower bound from \cite{mnp-lblsh-07}, and is
based crucially on an inequality proved there.

We mention that, in \cite{d-bcits-10}, Dubiner has also considered the
setting of a random dataset for a related problem---finding the
closest pair in a given dataset $P$. Dubiner sets up a certain
related ``bucketing'' model, in which he conjectures the lower bound,
which would imply a $\rho\ge \tfrac{1}{2c-1}$ lower bound for
data-independent LSH for a random set. Dubiner verifies the conjecture
computationally and claims it is proved in a different
manuscript.\footnote{This manuscript does not appear to be available
  at the moment of writing of the present paper.}

We also point out that, for the $\ell_2$ case, the optimal
data-independent lower bound $\rho\ge \frac{1}{2c^2 - 1} - o(1)$
follows from a recent work~\cite{practicalballcarving}. In fact, it
shows almost exact trade-off between~$p_1$ and $p_2$ (not only the
lower bound on $\rho = \frac{\log(1 / p_1)}{\log(1 / p_2)}$).
Unfortunately, the techniques there are really tailored to the
Euclidean case (in particular, a powerful isoperimetric inequality of
Feige and Schechtman is used~\cite{FS02}) and it is unclear how to
extend it to $\ell_1$, as well as, more generally, to~$\ell_p$ for
$1 \leq p < 2$.

Our second component focuses on the data-dependent aspect of the lower
bound. In particular, we prove that if there exists a data-dependent
hashing scheme for a random dataset with a better $\rho$, then in fact
there is also a data-independent such hashing scheme. To accomplish
this, we consider the ``empirical'' average $p_1$ and $p_2$ for a
random dataset, and prove it is close to the average $p_1$ and $p_2$,
for which we can deduce a lower bound from the first component.

In terms of related work, we also must mention the papers
of \cite{ptw-galba-08, ptw-lbnns-10}, who prove cell-probe lower
bounds for the same problem of ANN. In particular, their work utilizes
the lower bound of \cite{mnp-lblsh-07} as well. Their results are
however incomparable to our results: while their results are
unconditional, our model allows us to prove a much higher lower bound,
in particular matching the best algorithms.

\section{Data-independent Lower Bound}

In this section we prove the first component of Theorem
\ref{main_thm}. 
Overall we show a lower bound of $\rho \geq \frac{1}{2c - 1} + o(1)$
for data-independent hash families for random datasets. Our proof is a
strengthening of~\cite{mnp-lblsh-07}. The final statement appears as
Corollary~\ref{str_mnp}. In the second component, we will use a
somewhat stronger statement, Lemma~\ref{dubiner}).

For $u \in \{0, 1\}^d$ and non-negative integer $k$ define a random
variable $W_k(u)$ distributed over $\{0, 1\}^d$ to be the resulting
point of the standard random walk of length $k$ that starts in~$u$ (at
each step we flip a random coordinate).

We build on the following inequality from~\cite{mnp-lblsh-07} that is
proved using Fourier analysis on $\{0, 1\}^d$.

\begin{lemma}[\cite{mnp-lblsh-07}]
  \label{mnp_main}
  For every hash function $h \colon \{0, 1\}^d \to \Zbb$ and every odd positive integer $k$ one has:
  $$
  \underset{\substack{u \sim \{0, 1\}^d \\ v \sim W_k(u)}}{\mathrm{Pr}}[h(u) = h(u)] \leq
  \underset{u, v \sim \{0, 1\}^d}{\mathrm{Pr}}[h(u) = h(u)]^{\frac{\exp(2k / d) - 1}{\exp(2k / d) + 1}}.
  $$
\end{lemma}

\aanote{sounds ok?}
The above inequality can be thought of as a lower bound on $\rho$
already. In particular, the left-hand-side quantity is probability of
collision of a point and a point generated via a random walk of length
$k$ from it. The right-hand side corresponds to collision of random
independent points, which are usually at distance $d/2$. One already
obtain a lower bound on $\rho$ by considering $k=d/2c$.

To strengthen the lower bound of \cite{mnp-lblsh-07}, we analyze
carefully the distance between the endpoints of a random walk in the
hypercube. In particular the next (somewhat folklore) lemmas show that
this distance is somewhat smaller than the (trivial) upper bound of
$k$.  Denote $X_k$ the distance $\|W_k(u) - u\|_1$. Note that the
distribution of $X_k$ does not depend on a particular starting point
$u$.

\begin{lemma}
  \label{walk_exp}
  For every $k$ one has
  $$
  \mathrm{E}[X_k] = \frac{d}{2} \cdot \left(1 - \left(1 - \frac{2}{d}\right)^k\right).
  $$
\end{lemma}
\begin{proof}
  We have that
  \begin{multline*}
  \mathrm{E}[X_k \mid X_{k-1} = t] = \mathrm{Pr}[X_k = t - 1] \cdot (t - 1) + \mathrm{Pr}[X_k = t+1] \cdot (t + 1)
  \\ = \frac{t}{d} \cdot (t - 1) + \left(1 - \frac{t}{d}\right) \cdot (t + 1) =
  \left(1 - \frac{2}{d}\right) \cdot t + 1.
  \end{multline*}
  Thus,
  $$
  \mathrm{E}[X_k] = \left(1 - \frac{2}{d}\right) \cdot \mathrm{E}[X_{k-1}] + 1.
  $$
  Since $\mathrm{E}[X_0] = 0$, we obtain that
  $$
  \mathrm{E}[X_k] = \sum_{i=0}^{k-1} \left(1 - \frac{2}{d}\right)^i = \frac{d}{2} \cdot \left(1 - \left(1 - \frac{2}{d}\right)^k\right).
  $$
\end{proof}

We can now prove that the value of $X_k$ indeed concentrates well
around the expectation, using concentration inequalities.

\begin{lemma}
  \label{walk_concentr}
  For every $k$ and every $t > 0$ one has
  $$
  \mathrm{Pr}\left[X_k \geq \mathrm{E}[X_k] + t \cdot \sqrt{k}\right] \leq e^{-t^2 / 4}.
  $$
\end{lemma}
\begin{proof}
  For $t \geq 1$ define $Y_t$ to be the index of the coordinate that
  got flipped at time $t$. Obviously, we have that $X_k$ is a
  (deterministic) function of $Y_1,\ldots Y_k$. Furthermore, changing
  one $Y_t$ changes $X_k$ by at most 2. Hence we can apply the
  McDiarmid's inequality to $X_k=f(Y_1,\ldots Y_k)$:
\begin{theorem}[McDiarmid's inequality]
Let $Y_1,\ldots Y_k$ be independent random variables and
$X=f(Y_1,\ldots Y_k)$ such that changing variable $Y_t$ only changes
the value by at most $c_t$. Then we have that
$$
\Pr[X\ge \E{}{X}+\eps]\le \exp\left(-\frac{\eps^2}{\sum_{i=1}^k c_i^2}\right).
$$
\end{theorem}

Hence we obtain that  $$
  \mathrm{Pr}\left[X_k \geq \mathrm{E}[X_k] + \eps\right] \leq \exp\left(-\frac{\eps^2}{4k}\right).
  $$
  Substituting $\eps = t \cdot \sqrt{k}$, we get the result.
\end{proof}

Note that choosing $k\approx \tfrac{d}{2}\cdot \ln\tfrac{c}{c-1}$ will
mean that distance $X_k=\|W_k(u)-u\|_1$ is now around $d/2c$, i.e., we
can actually use random walks longer than the ones considered in
\cite{mnp-lblsh-07}. Indeed, from Lemma~\ref{walk_exp} and
Lemma~\ref{walk_concentr} one can immediately conclude the following
corollary.

\begin{corollary}
  \label{walk_corr}
  Let $c > 1$ be a fixed constant. Suppose that $\gamma = \gamma(d) > 0$ is such that $\gamma = o(1)$
  as $d \to \infty$. Then, there exists $\alpha = \alpha(d)$ such that
  $$
  \alpha(d) = \frac{1}{2} \cdot \ln \frac{c}{c - 1} - o_{c,\gamma}(1)
  $$
  that satisfies the following:
  $$
  \mathrm{Pr}\left[X_{\alpha \cdot d} > \frac{d}{2c}\right] \leq 2^{-\gamma \cdot d}.
  $$
\end{corollary}

We are now ready to prove the main lemma of this section, which will
be used in the later section on data-dependent hashing. We need to
introduce two more definitions. We define
``average $p_1$'' as:
$$
\zeta(c, d, \alpha, h) := 
\underset{\substack{u \sim \{0, 1\}^d\\v \sim W_{\alpha \cdot d}(u)}}{\mathrm{Pr}}
\left[h(u) = h(v) \middle| \|u - v\|_1 \leq \frac{d}{2c}\right]
$$
for $c > 1$, positive integer $d$, $\alpha > 0$ and a hash function $h \colon \{0, 1\}^d \to \Zbb$.
Similarly, we define the ``average $p_2$'' as
$$
\eta(d, \beta, h) :=
\underset{u, v \sim \{0, 1\}^d}{\mathrm{Pr}}\left[h(u) = h(v), \|u - v\|_1 > \left(\frac12 - \beta\right) \cdot d\right]
$$
for positive integer $d$, $\beta > 0$ and a hash function $h \colon \{0, 1\}^d \to \Zbb$.

\begin{lemma}
  \label{dubiner}
  Let $c > 1$ be a fixed constant.
  Suppose that $\gamma = \gamma(d) > 0$ is such that $\gamma = o(1)$
  as $d \to \infty$. Then, there exist $\alpha = \alpha(d)$, $\beta = \beta(d)$ and $\rho = \rho(d)$ such that:
  \begin{align*}
    \alpha &= \frac{1}{2} \cdot \ln \frac{c}{c - 1} - o_{c,\gamma}(1),\\
    \beta &= o_{\gamma}(1),\\
    \rho &= \frac{1}{2c - 1} - o_{c, \gamma}(1)
  \end{align*}
  as $d \to \infty$ such that, for every $d$ and every hash function $h \colon \{0, 1\}^d \to \Zbb$,
  one has
  $$
  \zeta(c, d, \alpha, h) \leq \eta(d, \beta, h)^{\rho} + 2^{-\gamma \cdot d}.
  $$
\end{lemma}

\begin{proof}
  We use Lemma~\ref{walk_corr} to choose
  $
  \alpha = \frac{1}{2} \cdot \ln \frac{c}{c - 1} + o_{c, \gamma}(1)
  $ such that $\alpha \cdot d$ is an odd integer and
  \begin{equation}
    \label{tail1}
  \mathrm{Pr}\left[X_{\alpha \cdot d} \geq \frac{d}{2c}\right] < \frac{2^{-\gamma \cdot d}}{2}.
  \end{equation}
  We can choose $\beta = o_{\gamma}(1)$ so that
  \begin{equation}
    \label{tail2}
  \underset{u, v \sim \{0, 1\}^d}{\mathrm{Pr}}\left[\|u - v\| < \left(\frac{1}{2} - \beta\right) \cdot d\right] < \frac{2^{-\gamma \cdot d}}{2}.
  \end{equation}
  (This follows from the standard Chernoff-type bounds.)

  Now we apply Lemma~\ref{mnp_main} and get
  \begin{equation}
    \label{mnp_main_app}
    \underset{\substack{u \sim \{0, 1\}^d\\ v \sim W_{\alpha \cdot d}(u)}}{\mathrm{Pr}}\left[h(u) = h(v)\right] \leq \underset{u, v \sim \{0, 1\}^d}{\mathrm{Pr}}\left[h(u) = h(v)\right]^{\rho},
  \end{equation}
  where
  $$
  \rho = \frac{\exp(2 \alpha) - 1}{\exp(2 \alpha) + 1} = \frac{1}{2c - 1} - o_{c, \gamma}(1).
  $$
  Finally, we combine~(\ref{mnp_main_app}), (\ref{tail1}) and~(\ref{tail2}), and get the desired inequality.
\end{proof}

The following corollary shows how the above lemma implies a lower
bound on data-independent LSH.

\begin{corollary}
  \label{str_mnp}
  For every $c > 1$ and $\gamma = \gamma(d) > 0$ such that $\gamma = o(1)$ there exists $\beta = \beta(d) > 0$ with
  $\beta = o_{\gamma}(1)$ such that
  if $\mathcal{H}$ is a data-independent $(d / (2c), (1/2 - \beta) d, p_1, p_2)$-sensitive family, then
  $$
  p_1 \leq p_2^{\frac{1}{2c - 1} - o_{c, \gamma}(1)} + 2^{-\gamma \cdot d}.
  $$
\end{corollary}
\begin{proof}
  We observe that for every $\alpha, \beta > 0$:
  \begin{itemize}
  \item $p_1 \leq \underset{h \sim \mathcal{H}}{\mathrm{E}}\left[\zeta(c, d, \alpha, h)\right]$;
  \item $p_2 \geq \underset{h \sim \mathcal{H}}{\mathrm{E}}\left[\eta(d, \beta, h)\right]$.
  \end{itemize}
  Now we apply Lemma~\ref{dubiner} together with the following application of Jensen's inequality:
  $$
  \underset{h \sim \mathcal{H}}{\mathrm{E}}\left[\eta(d, \beta, h)^{\rho}\right]
  \leq 
  \underset{h \sim \mathcal{H}}{\mathrm{E}}\left[\eta(d, \beta, h)\right]^{\rho},
  $$
  since $0 < \rho \leq 1$.
\end{proof}

\section{Data-Dependent Hashing}

We now prove the second component of the main Theorem \ref{main_thm}
proof. In particular we show that a very good data-dependent hashing
scheme would refute Lemma~\ref{dubiner} from the previous section.

\subsection{Empirical Probabilities of Collision}

For a particular dataset $P$, we will be interested in {\em empirical}
probabilities $p_1, p_2$ --- i.e., the equivalents of $\zeta, \mu$ for
a given set $P$ --- defined as follows. Let $0 < \delta(d) < 1/3$ be
some function. Let $P$ be a random set of points from $\{0, 1\}^d$ of
size $2^{\delta(d) \cdot d}$. The {\em empirical} versions of $\zeta$
and $\eta$ with respect to $P$ are:
\begin{align*}
\widehat{\zeta}(c, d, \alpha, h, P) & := 
\underset{\substack{u \sim P\\v \sim W_{\alpha \cdot d}(u)}}{\mathrm{Pr}}
\left[h(u) = h(v) \middle| \|u - v\|_1 \leq \frac{d}{2c}\right]\\
\widehat{\eta}(d, \beta, h, P) & :=
\underset{u, v \sim P}{\mathrm{Pr}}\left[h(u) = h(v), \|u - v\|_1 > \left(\frac12 - \beta\right) \cdot d\right].
\end{align*}

We now what to prove that, for a random dataset $P$, the empirical
$\zeta, \mu$ are close to the true averages. For this we will need the
following auxiliary lemma.

\begin{lemma}
  \label{conc2}
  Let $M$ be an $n \times n$ symmetric matrix with entries from $[0; 1]$ and average $\eps$.
  Let $M'$ be a principal $n^\delta \times n^\delta$ submatrix of $M$ sampled uniformly with replacement.
  Then, for every $\theta > 0$, the probability that the maximum of the average over $M'$ and $\theta$
  does not lie in $[1 / 2; 2] \cdot \max\{\eps, \theta\}$ is at most
  $n^{\delta} \cdot 2^{-\Omega(\theta n^{\delta})}$.
\end{lemma}
\begin{proof}
We need the following version of Bernstein's inequality.
\irnote{Find reference!}
\begin{lemma}
  \label{conc1}
  Suppose that $X_1, \ldots, X_n$ are i.i.d.\ random variables that are distributed over $[0; 1]$.
  Suppose that $\mathrm{E}[X_i] = \eps$. Then, for every $0 < \theta < 1$, one has
  $$
  \mathrm{Pr}\left[\max\left\{\frac{1}{n} \sum_{i=1}^n X_i, \theta \right\} \in \left[\frac{1}{2}; 2\right] \cdot \max\{\eps, \theta\}\right] \geq 1 - 2^{-\Omega(\theta n)}.
  $$
\end{lemma}

We just apply Lemma~\ref{conc1} and take union bound over the rows
of $M'$.
\end{proof}

The following two lemmas are immediate corollaries of
Lemma~\ref{conc1} and Lemma~\ref{conc2}, respectively.

\begin{lemma}
  \label{conc1_c}
  For every $c > 1$, $\alpha > 0$, positive integer $d$, $\theta > 0$ and a hash function $h \colon \{0, 1\}^d \to \Zbb$, one has
  $$
  \underset{P}{\mathrm{Pr}}\left[\max\{\widehat{\zeta}(c, d, \alpha, h, P), \theta\}
 \in \left[1/2;2\right] \cdot \max\{\zeta(c, d, \alpha, h), \theta\}\right] \geq 1 - 2^{-\Omega\left(\theta \cdot 2^{\delta(d) \cdot d}\right)}.
  $$
\end{lemma}

\begin{lemma}
  \label{conc2_c}
  For every $\beta > 0$, positive integer $d$, $\theta > 0$ and a hash function $h \colon \{0, 1\}^d \to \Zbb$, one has
  $$
  \underset{P}{\mathrm{Pr}}\left[
    \max\{\widehat{\eta}(d, \beta, h, P), \theta\}
 \in \left[1/2;2\right] \cdot \max\{\eta(d, \beta, h), \theta\}\right] \geq 1 - 2^{\delta(d) \cdot d} \cdot 2^{-\Omega\left(\theta \cdot 2^{\delta(d) \cdot d}\right)}.
  $$
\end{lemma}

\subsection{Proof of Theorem~\ref{main_thm}}

We are finally ready to complete the proof of the main result,
Theorem~\ref{main_thm}.

Let us first assume that $p_1 = o(1)$ and then show how to handle
the general case.  Suppose that $n = |P| = 2^{\delta \cdot d}$. By the assumption of Theorem~\ref{main_thm},
$\delta = o(1)$. Let $\{h_1, h_2, \ldots, h_T\}$ be a set of hash functions. We can assume that $\frac{\log T}{p_1} \leq n^{1 - \Omega(1)}$,
since otherwise we are done. Let us fix $\gamma = \gamma(d) > 0$ such that $\gamma = o(1)$ and
$2^{-\gamma d} \ll p_1^{\omega(1)}$.
We can do this, since if $p_1 = 2^{-\Omega(d)}$, then $1 / p_1 = n^{\omega(1)}$, and the desired statement
is true. Then, by Lemma~\ref{dubiner}, there is $\alpha = \alpha(d)$ and $\beta = \beta(d) = o_{\gamma}(1)$ such that for every $1 \leq i \leq T$
\begin{equation}
  \label{for_all}
\zeta(c, d, \alpha, h_i) \leq \eta(d, \beta, h_i)^{\frac{1}{2c - 1} - o_{c, \gamma}(1)} + 2^{-\gamma d}.
\end{equation}

Let us choose $\theta > 0$ such that
$$
T \cdot 2^{\delta \cdot d} \ll 2^{\theta \cdot 2^{\delta \cdot d}}
$$
and $\theta \ll p_1$. We can do it since, by the above assumption, $\frac{\log T}{p_1} \leq n^{1 - \Omega(1)} = 2^{\delta \cdot d (1 - \Omega(1))}$.

Then, from Lemma~\ref{conc1_c} and Lemma~\ref{conc2_c} we get that,
with high probability over the choice of $P$, one has for every $1 \leq i \leq T$:
\begin{itemize}
\item $\max\{\widehat{\zeta}(c, d, \alpha, h_i),\theta\}\in [1/2;2] \cdot \max\{\zeta(c, d, \alpha, h_i),\theta\}$;
\item $\max\{\widehat{\eta}(d, \beta, h_i),\theta\}\in [1/2;2] \cdot \max\{\eta(d, \beta, h_i),\theta\}$.
\end{itemize}

Suppose these conditions hold and assume there exist a distribution $\mathcal{D}$ over $[T]$ such that the corresponding
hash family is $(d / 2c, (1/2 - \beta(d)) d, p_1, p_2)$-sensitive for $P$.
Then,
\begin{multline}
  \label{trunc_1}
p_1 \leq \underset{i \sim \mathcal{D}}{\mathrm{E}}[\widehat{\zeta}(c, d, \alpha, h_i)]
\leq \underset{i \sim \mathcal{D}}{\mathrm{E}}[\max\{\widehat{\zeta}(c, d, \alpha, h_i), \theta\}]
\leq 2 \cdot \underset{i \sim \mathcal{D}}{\mathrm{E}}[\max\{\zeta(c, d, \alpha, h_i), \theta\}]
\\ \leq 2 \cdot \underset{i \sim \mathcal{D}}{\mathrm{E}}[\zeta(c, d, \alpha, h_i)] + \theta.
\end{multline}
Similarly,
\begin{multline}
  \label{trunc_2}
p_2 \geq \underset{i \sim \mathcal{D}}{\mathrm{E}}[\widehat{\eta}(d, \beta, h_i)]
\geq \underset{i \sim \mathcal{D}}{\mathrm{E}}[\max\{\widehat{\eta}(d, \beta, h_i), \theta\}] - \theta
\geq \frac{1}{2} \cdot \underset{i \sim \mathcal{D}}{\mathrm{E}}[\max\{\eta(d, \beta, h_i), \theta\}] - \theta
\\ \geq \frac{1}{2} \cdot \underset{i \sim \mathcal{D}}{\mathrm{E}}[\eta(d, \beta, h_i)] - \theta.
\end{multline}

Averaging~(\ref{for_all}) and applying Jensen's inequality, we have
\begin{equation}
  \label{for_each}
\underset{i \sim \mathcal{D}}{\mathrm{E}}[\zeta(c, d, \alpha, h_i)]
\leq \underset{i \sim \mathcal{D}}{\mathrm{E}}[\eta(d, \beta, h_i)]^{\frac{1}{2c - 1} - o_{}(1)} + 2^{-\gamma(d) \cdot d}.
\end{equation}

Thus, substituting~(\ref{trunc_1}) and~(\ref{trunc_2}) into~(\ref{for_each})
$$
\frac{p_1 - \theta}{2} \leq \left(2 (p_2 + \theta)
\right)^{\frac{1}{2c - 1} - o(1)} + 2^{-\gamma d},
$$ which proves the theorem, since $\theta \ll p_1$,
$2^{-\gamma d} \ll p_1^{\omega(1)}$, and $p_1 = o(1)$.

Now let us deal with the case $p_1 = \Omega(1)$. This can be
reduced to the case $p_1 = o(1)$ by choosing a slowly-growing
super constant $k$ and replacing the set of $T$ functions with the set
of $T^k$ \emph{tuples} of length $k$. This replaces $p_1$ and $p_2$
with $p_1^k$ and $p_2^k$, respectively. In the same time, we choose
$k$ so that $T' = T^k$ still satisfy the hypothesis of the theorem.  Then, we just
apply the above proof.

\singlespacing
{\small
\bibliographystyle{alpha}
\bibliography{desk}
}
\onehalfspacing
\appendix
\section{Upper Bounds Are in the Model}
\label{apx:upper}

In this section, we show how the data-dependent hash family
from~\cite{ar-optimal-15} fits into the model of the lower bound from
Section \ref{sec:model}.

Let us briefly recall the hash family construction
from~\cite{ar-optimal-15}.  For simplicity, assume that all points and
queries lie on a sphere of radius $R \gg cr$.  First, consider a
\emph{data-independent} hash family from~\cite{ainr-blsh-14,
  ar-optimal-15}, called \emph{Spherical LSH}.  It gives a good
exponent $\rho \leq \frac{1}{2c^2 - 1} + o(1)$ for distance thresholds
$r$ vs $\sqrt{2} R$ (the latter corresponds to a typical distance
between a pair of points from the sphere). The main challenge that
arises is how to handle distance thresholds $r$ vs $cr$, where the
latter may be much smaller than $\sqrt{2} R$. Here comes the main
insight of~\cite{ar-optimal-15}.

We would like to process the dataset so that the distance between a
typical pair of \emph{data points} is around $\sqrt{2} R$, so to apply
Spherical LSH. To accomplish this, we remove all the clusters of
radius $(\sqrt{2} - \eps) R$ that contain lots of points (think of
$\eps > 0$ being barely sub-constant). We will treat these clusters
separately, and will focus on the remainder of the pointset for
now. So for the remainder of the pointset, we just apply the Spherical
LSH to it. This scheme turns out to satisfy the definition of the
data-dependent hash family for the remaining points and for distance
thresholds $r$ vs. $cr$; in particular, the hash function is sensitive
for the remaining set only! The intuition is that, for any potential
query point, there is only a small number of data points within
distance $(\sqrt{2} - \eps)R$ --- otherwise, they would have formed
yet another cluster, which we would have removed --- and the larger
distances are handled well by the Spherical LSH. Thus, \emph{for a
  typical pair of data points}, Spherical LSH is sensitive for $P$
(see Definition~\ref{data_dependent_def}). 

How does \cite{ar-optimal-15} deal with the clusters? The main
observation is that one can enclose such a cluster in a ball of radius
$(1 - \Omega(\eps^2))R$, which intuitively makes our problem a little
bit simpler (after we reduce the radius enough times, the problem
becomes trivial). To answer a query, we query \emph{every cluster}, as
well as \emph{one part} of the remainder (partitioned by the Spherical
LSH).  This can be shown to work overall, in particular, we can
control the overall branching and depth of the recursion
(see~\cite{ar-optimal-15} for the details).

For both the clusters, as well as the parts obtained from the
Spherical LSH, the algorithm recurses on the obtained point
subsets. The overall partitioning scheme from~\cite{ar-optimal-15} can
be seen as a tree, where the root corresponds to the whole dataset,
and every node either corresponds to a cluster or to an application of
the Spherical LSH. One nuance is that, in different parts of the
tree the Spherical LSH partitioning obtains different $p_1, p_2$
(depending on $R$). Nonetheless, each time it holds that $p_1\ge
p_2^\rho$ for $\rho \leq \frac{1}{2c^2 - 1} + o(1)$. Hence, a node
terminates as a leaf once its accumulated $p_2$ (product of $p_2$'s of
``Spherical LSH'' nodes along the path from the root) drops below
$1/n$.

We now want to argue that the above algorithm can be recast in the
framework of data-dependent hashing as per
Definition~\ref{data_dependent_def}.  We consider a subtree of the
overall tree that contains the root and is defined as follows. Fix a
parameter $l=n^{-o(1)}$ (it is essentially the target $p_2$ of the
partition). We perform DFS of the tree and cut the tree at any
``Spherical LSH'' node where the cumulative $p_2$ drops below $l$.
This subtree gives a partial partition of the dataset $P$ as follows:
for ``Spherical LSH'' nodes we just apply the corresponding partition,
and for cluster nodes we ``carve'' them in the random order.  It turns
out that if we choose $l=n^{-o(1)}$ carefully, the partition will
satisfy Definition \ref{data_dependent_def} and the preconditions of
Theorem \ref{main_thm}. In particular, the description complexity of
the resulting hash function is $n^{o(1)}$.

Let us emphasize that, while Definition~\ref{data_dependent_def} is
certainly necessary for an ANN data structure based on data-dependent
hashing, it is not \emph{sufficient}. In fact, \cite{ar-optimal-15}
prove additional properties of the above partitioning scheme,
essentially because the ``$p_2$ property'' is ``on average'' one (thus, we
end up having to understand how this partitioning scheme treats
\emph{triples} of points).

\end{document}